\documentclass[commentary]{elsarticle}
\usepackage{bm}
\usepackage{hyperref,amsmath,amssymb}
\usepackage{graphicx}
\usepackage{natbib,amssymb}
\usepackage{graphicx}
\usepackage{hyperref}
\usepackage{amssymb}
\usepackage{amsthm}
\usepackage{yfonts}
\usepackage[inline]{enumitem}
\usepackage[final]{changes}
\newtheorem{definition}{Definition}[section]
\newtheorem{proposition}{Proposition}[section]
\newtheorem{corollary}{Corollary}[section]

\journal{tbd}
\newcommand{\y}{{\bf{y}}}
\newcommand{\x}{{\bf{x}}}
\newcommand{\X}{{\bf{X}}}
\renewcommand{\P}{{\bf{P}}}
\begin{document}

\begin{frontmatter}

\title{A sufficient condition for penalized polynomial regression to be invariant to translations of the predictor variables}

\author{Johannes W. R. Martini}
\ead{jwrmartini@gmail.com}

%\author[mymainaddress]{}

\begin{abstract} 
Whereas translating the coding of predictor variables does not change the fit of a polynomial least squares regression, penalized polynomial regressions are potentially affected. 
A result on which terms can be penalized to maintain the invariance to translations of the coding has earlier been published. A generalization of a corresponding proposition, which requires a more precise mathematical framework, is presented in this short note.

\end{abstract}


\end{frontmatter}

%\linenumbers

\section{Introduction}
Linear regressions are everyday tools in statistical applications. A particular type of regression, which is linear in the coefficients, and which allows for modeling interaction between the predictor variables, is polynomial regression. Here, the response $y$ is a polynomial in the predictors $\x=(x_1,...,x_p)$ 
\begin{equation}
	y_j= \sum\limits_{(i_1,i_2,...,i_p)\in I}a_{i_1,i_2,...,i_p} \; x_{j,1}^{i_1}x_{j,2}^{i_2}\cdot \cdot \cdot x_{j,p}^{i_p} + \epsilon_j, \qquad i_1,...,i_p \in \mathbb{N}
\end{equation}
In the standard setup, the data $\y=(y_1,...,y_n)$ and $\X=(x_{j,i})_{j=1,...n;i=1,...,p}$ is given, $\epsilon_j$ is assumed to be normally distributed $\epsilon_j \stackrel{i.i.d.}{\sim} \mathcal{N}(0,\sigma^2_\epsilon)$, and the coefficients $a_{i_1,i_2,...,i_p}$ have to be determined by a regression. The index set $I$ is chosen and defines which monomials are included in the model. \added{Here, $\X$ represents the matrix of data, but not the regressors of the regression problem, that is it is not the design matrix. The latter is obtained from $\X$ by multiplying its corresponding columns, according to the set of monomials $I$ of the polynomial model. Each monomial gives a regressor. } \\

Provided that the polynomial model satisfies \added{a ``completeness condition"} \cite{martini2019lost}, an ordinary least squares fit remains invariant when the coding of $\X$ is translated. Contrarily, it has been illustrated that the results of penalized regressions are likely to be affected by translations of the coding \citep{he2016does,martini2017genomic}. \added{Exceptions of penalized polynomial regressions being invariant to translations of the variable coding are given by those only penalizing the size of coefficients of monomials of highest total degree, provided the model allows to adapt all coefficients of lower degree \citep{martini2019lost}.} A generalization of this result is presented in this short note.
\section{Recapitulation of the required mathematical background}
We will recapitulate some technical background. The first definition provides a partial order on monomials which will be used to define the greatest monomials of a polynomial afterwards.
To simplify the treatise, we assume whenever we talk about a regression that a (unique) solution exists. Moreover, think in the following of $\x$ as any initially chosen coding of the predictor variable.
\begin{definition}[A partial order on monomials] For two monomials 
	$$m_1:= x_{1}^{i_1}x_{2}^{i_2} ... x_{p}^{i_p} \mbox{ and }	m_2:= x_{1}^{k_1}x_{2}^{k_2} ... x_{p}^{k_p},$$ we call $m_2$ \emph{greater than or of equal size as} $m_1$, in symbols
	$m_2 \geq m_1$, if
	$$ k_l \geq i_l  \qquad \forall l \in \{1,...,p\}.$$
\end{definition}
Note that if $m_1 \geq m_2$ and $m_2 \geq m_1$, it follows that $m_2=m_1$. \\

We use this partial order to specify what a ``greatest'' monomial is.

\begin{definition}[Greatest monomial] We call a monomial $m_1$ of a polynomial $f$ a \emph{greatest monomial} of $f$, if $f$ does not possess a monomial $m_2 \neq m_1$ with a non-zero coefficient which is greater than $m_1$. For a model $\{ I | I \subset \mathbb{R}^{|I|} \}$ a monomial is called greatest if there is no greater monomial $m_2 \neq m_1$ whose tuple of exponents is an element of $I$.
\end{definition}

Having defined what a greatest monomial is, we define ``translation'', ``translation invariance'' and the sum of squared residuals. 

\begin{definition}[Translation of a vector and of a polynomial] Let $\X$ be an $n \times p$ matrix of $p$ predictor variables measured $n$ times for the respective data $\y=(y_1,...,y_n)$. For a given \added{$1 \times p$-vector} $\P$, we define the translation of the predictor variables $$T_{\P}(\X):= \X + \mathbf{1}_n^t \P.$$
\added{$ \mathbf{1}_n$ denotes here the $1 \times n$-vector with each entry equal to 1.}
	Analogously, we define the translation of a polynomial $f(\x)$
	$$T_{\P}(f(\x)):= f(\x+ \P)= \left( f \circ T_{\P}\right)(\x).$$
\end{definition}
The definition of the translation of a polynomial above has the obvious property  
\begin{equation}\label{obvious}\left [ T_{-\P} \circ f \right ] \circ T_{\P}= f \end{equation}

We quickly define what we mean by ``translation invariance".

\begin{definition}[Translation invariance] A regression method $R$ that maps the data $(\X,\y)$ to a function $R_{\X,\y} (\x)$ is called \emph{translation invariant} if and only if 
	\begin{equation}\label{invariance}
		R_{\X,\y} (\x)=R_{T_\P(\X),\y} (T_\P(\x))
	\end{equation}
	for any data $(\X,\y)$ and any translation vector $\P$.
	%$T_{-\P} \circ f \circ T_{\P}(x) = T_{-\P} \circ f(x+\P^t) = f(x + \P^t - \P^t)=f(x)$
\end{definition}
In words, the definition of translation invariance means that for a regression, the resulting fit mapping $\x$ to $y$ (see Eq.(1)) is \added{identical} when applying the regression on $\left(\X,\y \right)$ to obtain $R_{\X,\y} (\x)$ or when using $\left(T_\P(\X),\y\right)$ to obtain a function $R_{T_\P(\X),\y} (T_\P(\x))$ defined on the translated predictor variables $T_\P(\x)$. 
For an example to see that this is not always the case, see Table~1 of \cite{martini2019lost}.

\begin{definition} Let the data $(\X,\y)$ be given. The function \emph{sum of squared residuals} (SSR) maps a polynomial $f$ to a real, non-negative number by
	\begin{equation}\label{SSR}SSR_{\X,\y}(f) := \sum_{j = 1,...,n} (y_j - f(\X_{j,\bullet}))^2.\end{equation}
	Here, $\X_{j,\bullet}$ denotes the $j$-th row of $\X$.
\end{definition}
With these definitions we can come to the results.

\section{Results}
\begin{proposition}\label{prop:01} Let $f(\x)$ be a polynomial. Then for any data $(\X,\y)$ and translation vector $\P$, 
	\begin{equation}\label{SSR2} SSR_{\X,\y}(f)=SSR_{T_\P(\X),\y}(T_{-\P}(f)).\end{equation}
	Moreover, for any greatest monomial $m$ of $f$, the corresponding coefficient $a_m$ of $f$ and $\tilde{a}_m$ of $T_{-\P}(f)$ will be identical:
	\begin{equation}\label{eq:06}a_m = \tilde{a}_m. \end{equation}
\end{proposition} 

\begin{proof} Concerning $SSR$, use the definition in Eq.(\ref{SSR}). Moreover, expand $f(\x-\P)$ to receive the coefficients of $T_{-\P}(f)$. A greatest monomial of $f$ will be a greatest monomial of $T_{-\P}(f)$. Inserting $\x-\P$ in a greatest monomial and expanding it, gives the same coefficient for this monomial.
\end{proof}

\begin{corollary} Let $G$ be the set of greatest monomials of polynomial $f$. Moreover, let $L_{\X,\y}$ be a loss function of \added{the form}  
	\begin{equation}\label{GOF}L_{\X,\y} = g(SSR_{\X,\y}) + PEN(a_{\in G})\end{equation}
	with $PEN$ any penalty function only dependent on the greatest monomials of $f$, and $g:\mathbb{R}^+_0 \rightarrow \mathbb{R}^+_0$ any function. Then 
	\begin{equation}\label{central}L_{\X,\y}(f)=L_{T_\P (\X),\y}(T_{-\P} (f)).\end{equation}
\end{corollary} 
\begin{proof} Eqs.(\ref{SSR2}-\ref{GOF}). 
\end{proof}

\begin{corollary}\label{cor:02} Let us consider a polynomial regression method $R$ defined by minimizing a loss function $L_{\X,\y}$ \added{of form} (\ref{GOF}). Moreover, let $\mathcal{F}$ denote the set of polynomials across which we look for the one minimizing $L_{\X,\y}$ and let  $\forall \, f \in \mathcal{F}$ and $\forall \, \P \in \mathbb{R}^p$ be $T_\P(f) \in \mathcal{F}$. Then
	\begin{equation}\label{Finalresult}
		R_{T_\P(\X),\y} (T_\P(\x)) = \left[ T_{-\P} \circ R_{\X,\y} \right] \circ T_\P (\x) = R_{\X,\y}(\x)
	\end{equation}
	and thus Eq.(\ref{invariance}) is fulfilled which means the fit is invariant to translations of the coding of the predictor variables. 
\end{corollary} 

\begin{proof} The second equality of Eq.(\ref{Finalresult}) is true for any function $f$ as stated in Eq.(\ref{obvious}). What requires a little bit more explanation is the first equality. Remember that
	$R_{\X,\y} (\x)$ is a function in $\x$ minimizing $L_{\X,\y}(f(\x))$.
	Eq.(\ref{central}) states that $T_{-\P} \circ R_{\X,\y}$ is the polynomial in $T_\P(\x)$ which minimizes $L_{T_\P(\X),\y}(f)$, which means that
	$$R_{T_\P(\X),\y} (T_\P(\x)) = \left[ T_{-\P} \circ R_{\X,\y} \right] \circ T_\P (\x). $$ 
\end{proof}
The statement of Corollary~\ref{cor:02} is more general than the result provided by \cite{martini2019lost}, since it allows the penalty function to be defined on the coefficients of the greatest monomials of the polynomial model, and not only on those of highest total degree. 
Monomials of highest total degree are greatest, but not every greatest monomial is of highest degree. \added{In particular, this more general sufficient condition may also be necessary.}

\section{Conclusion and Outlook}
We illustrated that any regression method defined by minimizing a loss function which is the sum of a function of the sum of squared residuals (SSR), and a penalty function only depending on the coefficients of the greatest monomials of the polynomial model, is invariant to translations of the coding of the predictor variables. Moreover, Eq.(\ref{GOF}) can easily be generalized by substituting the SSR by another function defined on a different norm, since Eq.(\ref{SSR2}) will still hold. \added{Finally, note that it may be the case that --for a regression defined by a loss function of form Eq.~(\ref{GOF})-- it is also a necessary condition that the penalty function only depends on the coefficients of greatest monomials. A general proof of this conjecture, maybe with some additional minor restrictions on the stucture of the regression problem, remains to be found.}

\bibliography{mybibfile}

\end{document}